\documentclass{amsart}

\usepackage{amsmath,amsthm,amsfonts, amssymb,amscd}
\usepackage[all]{xy}

\newtheorem{theorem}{Theorem}[section]
\newtheorem{lemma}[theorem]{Lemma}
\newtheorem{example}[theorem]{Example}
\newtheorem{remark}[theorem]{Remark}
\newtheorem{proposition}[theorem]{Proposition}

\newcommand{\Sh}{\mbox{\rm Sh}}

\newcommand{\RIP}{\mbox{\rm RIP}}

\begin{document}
\title[Quantum Observables and Effect Algebras]{Quantum Observables and Effect Algebras}
\author[Anatolij Dvure\v{c}enskij]{Anatolij Dvure\v{c}enskij$^{1,2}$}
\date{}%
\maketitle
\begin{center}  \footnote{


The paper has been supported by the grant VEGA No. 2/0069/16 SAV
 and GA\v{C}R 15-15286S. }
Mathematical Institute,  Slovak Academy of Sciences\\
\v Stef\'anikova 49, SK-814 73 Bratislava, Slovakia\\
$^2$ Depart. Algebra  Geom.,  Palack\'{y} Univer.\\
17. listopadu 12, CZ-771 46 Olomouc, Czech Republic\\

E-mail: {\tt
dvurecen@mat.savba.sk}
\end{center}

\begin{abstract}
We study observables on monotone $\sigma$-complete effect algebras. We find conditions when a spectral resolution implies existence of the corresponding observable. The set of sharp elements of a monotone $\sigma$-complete homogeneous effect algebra is a monotone $\sigma$-complete subalgebra. In addition, we study compatibility in orthoalgebras.
\end{abstract}

\section{Introduction}

The notion of an observable is a basic mathematical tool for study of quantum mechanical measurements. It is well known that quantum mechanical events do not satisfy axioms of the Kolmogorov axiomatical model, see \cite{BiNe, Var}, and rather have a more general structure. Mathematical foundations of quantum mechanics use many important algebraic structures, like orthomodular lattices, orthomodular posets, or orthoalgebras.  The most important ones of them are two structures connected with a Hilbert space which model the so-called Hilbert space quantum mechanics. So, if $H$ is a real, complex or quaternionic Hilbert space, let $\mathcal L(H)$ be the system of closed subspaces of $H$, then $\mathcal L(H)$ is a complete orthomodular lattice which is isomorphic to the set $\mathcal P(H)$ of all orthogonal projections on $H$. Observables in this structures correspond to orthogonal projector-valued measures, or equivalently, to symmetric operators. These structures describe so-called yes-no events.  A more general structure is the system $\mathcal E(H)$ of Hermitian operators on $H$ that are between the zero operator $O$ and the identity operator $I$. It describes yes-no events as well as so-called fuzzy events, and observables correspond to POVMs (positive operator valued measures). Of course, $\mathcal P(H) \subset \mathcal E(H)$.

Effect algebras have appeared in the beginning of the Nineties, see \cite{FoBe}. They combine both sharp and unsharp modeling of quantum events. For them observables are defined as  $\sigma$-homomorphisms from the Borel $\sigma$-algebra into a monotone $\sigma$-complete effect algebra. Today there is a reach literature on observables, see e.g. \cite{DvPu, JPV1, JPV2, DvKu, CIJTP}, and others. Recently the so-called Olson order of observables was introduced, \cite{286}, which was inspired by an order of operators presented by Olson, see \cite{Ols}.

There are two very important classes of effect algebras: The first class consists of effect algebras with the Riesz decomposition property (RDP for short), i.e. a joint refinement of any two decompositions is possible. The second class is formed by homogeneous effect algebras, introduced in \cite{Jen}. It contains the first class as well as orthoalgebras and lattice ordered effect algebras. These effect algebras will be central in our reasoning.

The aim of the present paper is to extend our knowledge on observables. In particular, we show a one-to-one relation between observables and their spectral resolutions, or we describe orthoalgebraic skeleton of homogeneous monotone $\sigma$-complete effect algebras.

The paper is organized as follows. Section 2 gathers the basic notions of the theory of effect algebras. Section 3 defines spectral resolutions. We show when a spectral resolution completely determines an observable. In Section 4, we describe the set of sharp elements of a monotone $\sigma$-complete homogeneous effect algebra; sharp elements will form the orthoalgebraic skeleton of an effect algebra. We note that $\mathcal P(H)$ is the orthoalgebraic skeleton of $\mathcal E(H)$. In addition, we show when effect algebras satisfy the Observable Existence Property.

\section{Elements of Effect Algebras}

According to \cite{FoBe}, we say that an {\it effect algebra} is a partial algebra $E = (E;+,0,1)$ with a partially defined operation $+$ and with two constant elements $0$ and $1$  such that, for all $a,b,c \in E$,
\begin{enumerate}

\item[(i)] $a+b$ is defined in $E$ if and only if $b+a$ is defined, and in
such a case $a+b = b+a;$

 \item[(ii)] $a+b$ and $(a+b)+c$ are defined if and
only if $b+c$ and $a+(b+c)$ are defined, and in such a case $(a+b)+c
= a+(b+c)$;

 \item[(iii)] for any $a \in E$, there exists a unique
element $a' \in E$ such that $a+a'=1$;

 \item[(iv)] if $a+1$ is defined in $E$, then $a=0$.
\end{enumerate}

If we define $a \le b$ if and only if there exists an element $c \in
E$ such that $a+c = b$, then $\le$ is a partial ordering on $E$, and
we write $c:=b-a$. It is clear that $a' = 1 - a$ for all $a \in E$, and if $a\le b$, then $b-a=(b'+a)'$. For more information about effect algebras, see \cite{DvPu}.  A {\it homomorphism}  from an effect algebra $E_1$ into another one $E_2$ is any mapping $h: E_1 \to E_2$ such that  (i) $h(1)=1$ (ii) if $a+b$ is defined in $E_1$ so is defined $h(a)+h(b)$ in $E_2$ and $h(a+b)= h(a)+h(b)$. A subset $F$ of an effect algebra $E$ is a {\it sub-effect algebra} of $E$ if (i) $0,1 \in F,$ (ii) if $a\in F,$ then $a'\in F$, and (iii) if $a,b\in F$ and $a+b \in E$, then $a+b \in F$.

If $G$ is an Abelian partially ordered group written additively, choose an element $u \in G^+:=\{g \in G \colon g \ge 0\}$, and set $\Gamma(G,u):=[0,u]=\{g \in G: 0 \le g \le u\}$. Then $(\Gamma(G,u);+,0,u)$ is an effect algebra, where $+$ is the group addition of elements from $\Gamma(G,u)$  if it exists in $\Gamma(G,u)$. Effect algebras that are isomorphic to some $\Gamma(G,u)$ are said to be {\it interval effect algebras}. For more information on Abelian partially ordered groups see \cite{Goo}. A sufficient condition to be an effect algebra an interval one is the Riesz Decomposition Property, \cite{Rav}. We say that an effect algebra $E$ satisfies the Riesz Decomposition Property (RDP for short), if $a_1+a_2=b_1+b_2$ implies that there are four elements $c_{11},c_{12},c_{21},c_{22}\in E$ such that $a_1 = c_{11}+c_{12},$ $a_2= c_{21}+c_{22},$ $b_1= c_{11} + c_{21}$ and $b_2= c_{12}+c_{22}$. Equivalently, $E$ has RDP iff $a\le b+c$ implies that there are $b_1,c_1\in E$ such that $b_1\le b$, $c_1\le c$ and $a=b_1+c_1$.
An effect algebra $E$  has the {\it Riesz Interpolation Property} (RIP for short) if, given $x_1,x_2,y_1,y_2$ in $E$ such that  $x_i \le y_j$ for all $i,j=1,2$, there exists an element $z \in E$ such that $x_i \le z\le  y_j$ for all $i,j=1,2$. We note that RDP implies RIP, but the converse is not always true.

An important class of effect algebras consists of clans which are families of functions where all algebraic operations are defined by points. A {\it clan} is a family $\mathcal C$ of fuzzy sets  (= $[0,1]$-valued functions) on a set $\Omega \ne \emptyset$ such that (i) $1\in \mathcal C$, (ii) if $f \in {\mathcal C}$, then $1-f \in {\mathcal C}$, (iii) if $f,g \in {\mathcal C}$, $f \le 1-g$, then $f+g \in {\mathcal C}$.

An effect algebra $E$ is {\it monotone} $\sigma$-{\it complete} if, for any sequence $a_1 \le a_2\le \cdots,$ the element $a = \bigvee_n a_n$  is defined in $E$ (we write $\{a_n\}\nearrow a$). If an effect algebra is a lattice or a $\sigma$-lattice or a complete lattice, we say that $E$ is a {\it lattice effect algebra}, a $\sigma$-{\it lattice effect algebra}, and a {\it complete lattice effect algebra}, respectively.

If $E$ and $F$ are two monotone $\sigma$-complete effect algebras, a homomorphism $h:E \to F$ is said to be a $\sigma$-{\it homomorphism} if $\{a_n\} \nearrow a$ implies $\{h(a_n)\} \nearrow h(a)$ for $a, a_1,\ldots \in E$.

Interesting examples of monotone $\sigma$-complete effect algebras are effect-tribes. We remind that  an {\it effect-tribe}  is any system ${\mathcal T}$ of $[0,1]$-valued functions on
$\Omega\ne \emptyset $ such that (i) $1 \in {\mathcal T}$, (ii) if $f
\in {\mathcal T}$, then $1-f \in {\mathcal T}$, (iii) if $f,g \in {\mathcal T}$,
$f \le 1-g$, then $f+g \in {\mathcal T}$, and (iv) for any sequence $\{f_n\}$ of elements of ${\mathcal T}$ such that $f_n \nearrow f$ (pointwise), then $f \in {\mathcal T}$. It is evident that any effect-tribe is a monotone $\sigma$-complete effect algebra where all algebraic operations are defined by points. We note that $\mathcal E(H)$ is isomorphic to the effect-tribe $\mathcal T(H)=\{(A\phi,\phi): \phi\in H, \|\phi\|=1, A \in \mathcal E(H)\}$ and analogously,  the same is true for $\mathcal P(H)$.

An effect algebra $E$ is an {\it orthoalgebra} if the existence of $a+a$ implies $a=0$. An orthoalgebra $E$ is an {\it orthomodular poset} (OMP for short) if the existence of $a+b$ implies $a\vee b$ exists and in such a case $a+b=a\vee b$, \cite[Prop 1.5.6]{DvPu}. An orthoalgebra $E$ is an orthomodular poset iff the existence of $a+b, a+c, b+c$ implies $a+b+c$ is defined in $E$, \cite[Thm 1.5.5]{DvPu} or \cite[Thm 2.12]{FGR}. If an OMP is also a lattice, we call it an {\it orthomodular lattice} (OML for short).

A more general notion than an effect algebra with RDP is the following notion \cite{Jen}: We say that an effect algebra $E$ is {\it homogeneous} if, whenever $a,b,c \in E$ are such that $a \le b+c$, $a\le (b+c)'$, there are $a_1,a_2 \in E$ such that $a_1\le b,$ $a_2 \le c$ and $a=a_1 + a_2$. We notice that (1) every effect algebra with RDP is homogeneous, (2) every lattice effect algebra is homogeneous, (3) every orthoalgebra is homogenous, (4) every homogeneous effect algebra can be covered by a system of sub-effect algebras $(E_t\colon t \in T)$ of $E$ such that every $E_t$ satisfies RDP, \cite[Thm 3.1, Cor 3.13]{Jen}.

A finite sequence $(a_1,\ldots, a_n)$ of elements of $E$ is {\it summable} if $a:=a_1+\cdots+a_n:=\sum_{n=1}^na_i $ exists, and the element $a$ is said to be the {\it sum} of $(a_1,\ldots,a_n)$. An arbitrary family $(a_t: t \in T)$ of elements of $E$ is said to be {\it summable} if every finite subsystem of $(a_t: t \in T)$ is summable. If, in addition, there exists $a :=\bigvee\{\sum\{a_t: t \in S\}: S$ is a finite subset of $T\}$ in $E$, the element $a$ is said to be the {\it sum} of $(a_t: t \in T)$, and we write $a = \sum_{t \in T}a_t$.

Two elements $a$ and $b$ of an effect algebra $E$ are said to be (i) {\it compatible} and we write $a \leftrightarrow b$ if there
exist three elements $a_1, b_1, c \in E$ such that $a= a_1 + c$, $b= b_1 + c$ and $a_1 + b_1 + c \in E$, (ii) {\it strongly compatible} and we write $a\stackrel{\mbox{\rm c}}{\longleftrightarrow} b$ if there are three elements $a_1, b_1, c \in E$ such that $a = a_1 + c$, $b = b_1 + c$, $a_1 \wedge  b_1 = 0$ and $a_1 +b _1 + c \in E$. We note that $a\stackrel{\mbox{\rm c}}{\longleftrightarrow}b$ implies $a\leftrightarrow b$, but the converse implication does not hold, in general.

We say that an arbitrary subset $M$ of $E$ is (i) {\it compatible} if, for any finite subset $M_F$ of $M,$ there is a summable sequence $(c_1,\ldots,c_k)$  of elements of $E$ such that every element of $M_F$ can be expressed as a sum of some elements from $\{c_1,\ldots,c_k\}$;
(ii) {\it internally compatible} if, for any finite subset $M_F$ of $M,$ there is a summable sequence $(c_1,\ldots,c_k)$  of elements of $M$ such that every element of $M_F$ can be expressed as a sum of some elements from $\{c_1,\ldots,c_k\}$.

A {\it block} of an effect algebra $E$ is a maximal sub-effect algebra, $B$, of $E$ (if it exists in $E$) satisfying RDP. If $E$ is a homogeneous effect algebra, then a subset $B$ of $E$ is a block iff $B$ is a maximal internally compatible set with $1\in B$, \cite[Thm 3.11]{Jen}, and $E$ can be covered by its blocks, \cite[Cor 3.13]{Jen}.

A {\it state} on an effect algebra $E$ is any mapping $s: E \to [0,1]$ such that (i) $s(1)=1$, and (ii) $s(a+b)=s(a)+s(b)$ whenever $a+b$ is defined in $E$. A state $s$ is $\sigma$-{\it additive} if (iii) if $a_n \nearrow a$, then $s(a)=\lim_n s(a_n)$. Equivalently, a state $s$ is $\sigma$-additive if $a= \sum_n a_n$ for a sequence of summable elements $\{a_n\}$, then $s(a)=\sum_n s(a_n)$. For every state $s$ we have (i) $s(a')=1-s(a)$, (ii) $a\le b$ implies $s(a)\le s(b)$.

A non-empty system of states $\mathcal S$ is (i) {\it order-determining} if $s(a)\le s(b)$ for each $s \in \mathcal S$ implies $a\le b$, (ii) {\it full} if, for any $a>0$, there is $s \in \mathcal S$ such that $s(a)=1$.
For example, if $E=\mathcal E(H)$ and $s_\phi(A):=(A\phi,\phi)$, $A \in \mathcal E(H)$, where $\phi$ is any unit vector in $H$, then $\mathcal S_v(H):=\{s_\phi: \phi \in H, \|\phi\|=1\}$ is an order-determining system of $\sigma$-additive states on $\mathcal E(H)$ whilst it is not full; take e.g. $A=1/2I$. A state $s_\phi$ is said to be a {\it vector state} or a {\it pure state}. The restriction of $s_\phi$ onto $\mathcal P(H)$ is also a $\sigma$-additive state and we will denote it also as $s_\phi$.  The system $\mathcal S_v(H)$ is for $\mathcal P(H)$ order-determining as well as full. Similarly, if $\mathcal T$ is an effect-tribe, $s_\omega(f):=f(\omega)$, $f \in \mathcal T$ ($\omega \in \Omega)$, then $\{s_\omega: \omega\in \Omega\}$ is an order-determining system of $\sigma$-additive states on $\mathcal T$ but not necessarily full.

We note that if an effect algebra $E$ has a full system of states, then $x+x \in E$ implies $x=0$. Indeed, $x+x$ means $x\le x'$. If $x>0$, there is a state $s$ such that $s(x)=1$ which entails $s(x)\le s(x')=1$, i.e. $s(x)=0$, a contradiction. In other words, an effect algebra with a full system of states is an orthoalgebra, \cite[Lem 1.5.2]{DvPu}.

An element of an effect algebra $E$ is said to be {\it sharp} if $a \wedge a'$ exists in $E$ and $a\wedge a'=0.$ Let $\Sh(E)$ be the set of sharp elements of $E$. Then (i) $0,1\in \Sh(E),$ (ii) if $a \in \Sh(E),$ then $a'\in \Sh(E).$ If $E$ is a lattice effect algebra, then $\Sh(E)$ is an orthomodular lattice which is a sub-effect algebra and a sublattice of $E$,  \cite{JeRi}. If an effect algebra $E$ satisfies RDP, then by \cite[Thm 3.2]{Dvu2}, $\Sh(E)$ is even a Boolean algebra. If $E$ is a homogenous effect algebra, then $\Sh(E)$ is a sub-effect algebra of $E$, \cite[Cor 4.4]{Jen}, and moreover, $\Sh(E)$ is an orthoalgebra.

An element $a\in E$ is said to be {\it principal} if $x,y \le a$ and $x+y \in E$ imply $x+y \le a$. A principal element $a\in E$ is said to be {\it central}, if for each $b \in E$, there are unique elements $b_1,b_2\in  E$ such that $b_1\le a$ and $b_2 \le a'$ and $b = b_1+b_2$. We denote by $C(E)$ the set of central elements of $E$. We note that an element $a$ of an effect algebra $E$ is central iff there exists an isomorphism
$f_a:\, E \to [0,a] \times [0,a']$ such that $f_a(a) =(a,0)$ and if $f_a(b) = (b_1,b_2)$, then $b=b_1 + b_2$ for any $b \in E$. Then $C(E)$ is a Boolean algebra such that (i) $0,1\in C(E)$, (ii) $a \in C(E)$ implies $a' \in C(E)$, and (iii) for each $a \in C(E)$ and each $x \in E$, $a\wedge x$ and $a\vee x$  exist in $E$, and $(x\vee a)-x=a- (a\wedge x)$,   and (iv) if $E$ is monotone $\sigma$-complete, then $C(E)$ is a Boolean $\sigma$-algebra, \cite[Thm 5.11]{Dvu2}.

We note that $C(\mathcal E(H))=\{O,I\}=C(\mathcal P(H))$ and $\Sh(\mathcal E(H))=\mathcal P(H)$, \cite[Thm 4.4]{104}. In addition, $\mathcal P(H)$ is a sub-effect algebra of $\mathcal E(H)$ and $\mathcal P(H)$ is a complete (orthomodular) lattice such that suprema and infima of projections on $H$ are the same in both effect algebras. We note that $\mathcal P(H)$ is the greatest sub-orthoalgebra (sub-orthomodular poset or sub-orthomodular lattice) of $\mathcal E(H)$ containing $\mathcal P(H)$. Moreover, RDP fails in $\mathcal E(H)$ as well as in $\mathcal P(H)$, and $\mathcal E(H)$ is homogeneous iff $\dim H\le 1$, \cite[Cor 5.5]{Jen}.

\section{Spectral Resolution and Observables} 

In the section, we define observables, and we point out to their very close connection with spectral resolutions. We show when a spectral resolution determines uniquely an observable. We introduce the Observable Existence Property for effect algebras; it guarantees the existence of an observable corresponding to a given spectral resolution.

Let $E$ be a monotone $\sigma$-complete effect algebra. An {\it observable} on $E$ is any mapping $x:\mathcal B(\mathbb R)\to E$, where $\mathcal B(\mathbb R)$ is the Borel $\sigma$-algebra of the real line $\mathbb R$, such that (i) $x(\mathbb R)=1$, (ii) if $A,B \in \mathcal B(\mathbb R)$, $A \cap B= \emptyset$, then $x(A\cup B)=x(A)+x(B)$, and (iii) if $\{A_i\}$ is a sequence of Borel sets such that $E_i\subseteq E_{i+1}$ for each $i$ and $\bigcup_i A_i=A$, then $x(A) = \bigvee_i x(A_i)$. The set $\mathcal R(x)=\{x(A): A \in \mathcal B(\mathbb R)\}$ is said to be the  {\it range} of an observable $x$.

Observe that any observable is a $\sigma$-homomorphism of monotone $\sigma$-complete effect algebras. The basic properties of observables are
(i) $x(\mathbb R \setminus A)=x(A)'$, (ii) $x(\emptyset)=0$, (iii) $x(A)\le x(B)$ whenever $A \subseteq B$, and $x(B\setminus A)= x(B)-x(A)$, (iv) if $\{A_i\} \searrow A$, then $x(A) = \bigwedge_i x(A_i)$, (v) if $x(A)=x(B)=0$, then $x(A\cup B)=x(A\setminus B)+x(B\setminus A)+x(A\cap B)=0$,  (vi) if $x(A_n)=0$ for each $n\ge 1$, then $x(\bigcup_n A_n)=0$, (vii) $x(A)+x(B)$ exists iff does $x(A\cup B)+x(A\cap B)$, and in such a case,  $x(A)+x(B)=x(A\cup B)+x(A\cap B)$, (viii) $x(A\cup B)-x(A)=x(B)-x(A\cap B)$, (ix) if $x(B)=0$, then $x(A\cup B)=x(A)$, (x) if $x(B)=1$, then $x(A)=x(A\cap B)$, and (xi) $x(A) \leftrightarrow x(B)$. (xi) $\mathcal R(x)$ is an internally compatible subset of $E$. Indeed, given a Borel set $A$, we define $A^1:=A$ and $A^0:=\mathbb R \setminus A$. If $A_1,\ldots,A_n$ are Borel sets, for any sequence $(i_1,\ldots,i_n)$ of $0$'s and $1$'s, let $A(i_1,\ldots,i_n):= A_1^{i_1}\cap\cdots \cap A_n^{i_n}$. Then $\{x(A(i_1,\ldots,i_n)): (i_1,\ldots,i_n)\in \{0,1\}^n\}$ is a summable system of elements from $\mathcal R(x)$ and each element $x(A_i)$ is a sum of some finitely many elements of the summable system.

The least closed subset $C$ of $\mathcal B(\mathbb R)$ is said to be the {\it spectrum} of $x$, and we denote it by $\sigma(x)$; since the natural topology of the real line satisfies the second countability axiom, $\sigma(x)$ exists, and $x(\sigma(x))=1$. In addition, a point $\lambda \in \mathbb R$ belongs to $\sigma(x)$ iff for any open set $U$ containing $\lambda$, $x(U)\ne 0$. An observable $x$ is: (i) {\it bounded}, if $\sigma(x)$ is a compact set, (ii) {\it bounded from above} ({\it from below}), if there is $k\in \mathbb R$ such that $\sigma(x) \subseteq (-\infty,k]$ ($\sigma(x) \subseteq [k,\infty)$), (iii) {\it positive} ({\it negative}) if $\sigma(x) \subseteq [0,\infty)$ ($\sigma(x) \subseteq (-\infty,0])$, (iv) a {\it question} if $\sigma(x)\subseteq \{0,1\}$, (v) {\it simple} if $\sigma(x)$ is a finite non-empty subset of $\mathbb R$, and (vi) an {\it effect observable} if $\sigma(x) \subseteq [0,1]$.

Given a bounded observable $x$, we define the {\it norm} of $x$ as $\|x\| := \sup\{|\lambda|: \lambda \in \sigma(x)\}$. Then $x=q_0$ iff $\|x\|=0$.

If $x$ is an observable and $f: \mathbb R\to \mathbb R$ is a Borel measurable function, then the mapping $f\circ x: A \mapsto x(f^{-1}(A))$, $A \in \mathcal B(\mathbb R)$, is also an observable. If, e.g. $f(t)=t^n$, $t \in \mathbb R$, then we write $f\circ x = x^n$. It is possible to show that if $x$ is a bounded observable and $f$ is continuous on $\sigma(x)$, then $\sigma(f\circ x)=f(\sigma(x))$.

We denote by $\mathcal O(E)$, $\mathcal {BO}(E)$ and $\mathcal {EO}(E)$ the set of all observables, bounded observables and the set of effect observables, respectively, on $E$.

For example, let $\{a_n\}$ be a finite or infinite sequence of summable elements, $\sum_n a_n = 1$, and let $\{t_n\}$ be a sequence of mutually different real numbers. Then the mapping $x: \mathcal B(\mathbb R) \to E$ defined by

$$
x(A):= \sum\{a_n: t_n \in A\}, \ A \in \mathcal B(\mathbb R), \eqno(3.0)
$$
is an observable on $E$. In particular, if $t_0=0$, $t_1=1$ and $a_0=a'$ $a_1=a$ for some fixed element $a \in E,$  $x$ defined by (3.0) is an observable, called a {\it question corresponding to the element} $a$, and we write $x=q_a$.

The following characterization of  question observables was known for observables on $\sigma$-orthomodular lattices, see e.g. \cite[Lem 3.5(iii)]{Gud}. Here we generalize it as follows:

\begin{lemma}\label{le:3.1}
A bounded observable $x$ on a monotone $\sigma$-complete effect algebra $E$ is a question observable if and only if $x=x^2$.
\end{lemma}

\begin{proof}
Clearly, if $\sigma(x)\subseteq \{0,1\}$, then $x=x^2$.  Conversely, assume that $x=x^2$. Then $x$ is positive, so that $x((-\infty,0))=0$. Choose a rational $t>1$. Then $t< t^2$ so that $x([t^2,\infty))= x^2([t^2,\infty))= x([t,\infty))$ which yields $x([t,t^2))= x([t^2,\infty)) - x([t,\infty))=0$. In a similar way, we have $x([t^{2^n},t^{2^{n+1}}))=0$ for each $n\ge 1$. Hence, $x([t,\infty))=\sum_{n=0}^\infty x([t^{2^n},t^{2^{n+1}})=0$, and since $t>1$ was arbitrary, we have $x((0,\infty))=\bigvee_{t\searrow 0}x((t,\infty))=0$.

Now let a rational number $t$ be such that $0<t<1$. Then $t^2<t$, so that $x([t^2,t))=0$, and in a similar way, we have $x([t^{2^{n+1}},t^{2^n}))=0$ for each $n\ge 1$. Therefore, $x((0,t))=0$. Since $t<1$ was arbitrary, we have $x((0,1))=0$.

Finally, we have $\sigma(x)\subseteq \{0,1\}$.
\end{proof}

We present a simple but useful proposition, see also \cite[Rem 3.3]{DvKu}:

\begin{proposition}\label{pr:3.2}
Let $x$ be an observable on a monotone $\sigma$-complete effect algebra $E$. Given a real number $t \in \mathbb R,$ we put

$$ x_t := x((-\infty, t)). \eqno(3.1)
$$
Then

$$ x_t \le x_s \quad {\rm if} \ t < s, \eqno (3.2)$$

$$\bigwedge_t x_t = 0,\quad \bigvee_t x_t =1, \eqno(3.3)
$$
and
$$ \bigvee_{t<s}x_t = x_s, \ s \in \mathbb R. \eqno(3.4)
$$

Conversely, if there is a system $\{x_t: t \in \mathbb R\}$ of elements of $E$ satisfying {\rm
(3.2)--(3.4)} and if there is an observable $x$ on $E$ for which $(3.1)$ holds for any $t \in \mathbb R$, then $x$ is unique.
\end{proposition}

\begin{proof}
Since $E$ is a monotone $\sigma$-complete effect algebra, from (3.2) we conclude that (3.4) holds. Indeed, if $\{t_n\}\nearrow s$ and $\{t'_n\}\nearrow s$ are two sequences, then $\bigvee_n x_{t_n}=x_s=\bigvee_n x_{t'_n} = \bigvee_{t<s} x_t$. Similarly for the other properties.

Now assume that $x$ is an observable such that $x((-\infty,t))=x_t$ for each $t \in \mathbb R$. Let $y$ be another observable of $E$ such that $y((-\infty,t))=x_t$ for each $t \in \mathbb R$. We denote by $\mathcal K=\{A \in \mathcal B(\mathbb R): x(A)=y(A)\}$. Then $\mathcal K$ is a  Dynkin system, i.e. a system of subsets containing its universe which is closed under the set theoretical complements and countable unions of  disjoint subsets, \cite{Bau}. It contains all intervals of the form $(-\infty,t)$ for each $t \in \mathbb R$; these intervals form a $\pi$-system, i.e. intersection of any two sets from the $\pi$-system is from the $\pi$-system.  Hence by the Sierpi\'nski Theorem, \cite[Thm 1.1]{Kal}, $\mathcal K$ is also a $\sigma$-algebra, and consequently, we have $\mathcal K= \mathcal B(\mathbb R)$.
\end{proof}

It is worthy to recall the following note. Let $\{x_t: t \in \mathbb R\}$ satisfy (3.2)--(3.4), $\mathbb R^*=\mathbb R \cup \{-\infty\}\cup \{\infty\}$, and $x_{-\infty}=0$ and $x_\infty =1$. Define $M=\{x_t-x_s: s,t \in \mathbb R^*, s\le t\}$. Then $M$ is internally compatible. (The same is true if we change $\mathbb R$ to $\mathbb Q$.) Indeed, let $F=\{x_{t_1}-x_{s_i},\ldots, x_{t_n}-x_{s_n}\}$, where $s_i\le t_i$ and $t_i,s_i\in \mathbb R^*$ for $i=1,\ldots,n$. Then $\{t_1,s_2,\ldots, t_n,s_n\}=\{u_1,\ldots,u_m\}$, where $u_1<\cdots <u_m$. Then $S=\{x_{u_1},x_{u_2}-x_{u_1},\ldots, x_{u_m}-x_{u_{m-1}}\}$ is a summable system of $E$ and every element $x_{t_i}-x_{s_i}$ for $i=1,\ldots,n$, is a sum of some finitely many elements of $S$.

The system $\{x_t: t \in \mathbb R\}$ from Proposition \ref{pr:3.2} satisfying (3.2)--(3.4) is said to be the {\it spectral resolution} of an observable $x$. To represent observables by their spectral resolutions and vice-versa, we have to establish conditions when a given system $\{x_t: t \in \mathbb R\}$ of elements of $E$ satisfying (3.2)--(3.4) determines an observable $x$ such that $x((-\infty,t))=x_t$, $t \in \mathbb R$. By Proposition \ref{pr:3.2}, it is enough to find at least one observable $x$, then it is unique. Some such conditions were established in \cite{DvKu, 270}.

Motivated by this problem, we say that a monotone $\sigma$-complete effect algebra $E$ has the {\it Observable Existence Property} (OEP, for short) if given a system of elements $\{x_t: t \in \mathbb R\}$  of $E$ satisfying (3.2)--(3.4), there is an observable $x$ such that $x((-\infty,t))=x_t$, $t \in \mathbb R$. We denote by $\mathcal{OEP}(EA)$ the class of effect algebras with OEP. The class $\mathcal{OEP}(EA)$ contains these effect algebras:

\begin{enumerate}
\item[(i)] $\sigma$-MV-algebras, \cite[Thm 3.2]{DvKu}.
\item[(ii)] $\sigma$-lattice effect algebras, \cite[Thm 3.5]{DvKu}.
\item[(iii)] Boolean $\sigma$-algebras, \cite[Thm 3.6]{DvKu}.
\item[(iv)] $\sigma$-orthocomplete orthomodular posets, \cite[Thm 3.8]{DvKu}.
\item[(v)] Monotone $\sigma$-complete effect algebras with RDP, \cite[Thm 3.9]{DvKu}.
\item[(vi)] $\mathcal E(H)$, \cite[Thm 3.10]{DvKu}.
\item[(vii)] Effect-tribes, \cite[Thm 3.11]{DvKu}.
\item[(viii)] Monotone $\sigma$-complete effect algebras with RIP and DMP, \cite[Thm 4.3]{270}.
\end{enumerate}

We do not know whether every monotone $\sigma$-complete effect algebra belongs to the class $\mathcal{OEP}(EA)$.

A partial answer to this problem is the following result. First we define a kind of representable monotone $\sigma$-complete effect algebras. Another type of representable effect algebras was defined in \cite{270}.

We say that a monotone $\sigma$-complete effect algebra $E$ is {\it representable} if there is an effect-tribe of fuzzy sets $\mathcal T$ and a $\sigma$-homomorphism $h$ of monotone $\sigma$-complete effect algebras from $\mathcal T$ onto $E$ such that if $f,g \in \mathcal T$ with $f \le g$ and if there is $c_1 \in E$ such that $h(f)\le c_1\le h(g)$, then there is $c\in \mathcal T$ such that $h(c)=c_1$ and $f\le c\le g$. For example, every (1) monotone $\sigma$-complete effect algebra with RDP (\cite{BCD}) and \cite[Thm 3.9, Claim 2]{DvKu}, (2) $E$ which is a $\sigma$-MV-algebra \cite{Dvu1}, (3) $\mathcal E(H)$, (4) effect-tribe is representable. Representable monotone $\sigma$-complete effect algebras are connected with the Loomis--Sikorski representation of monotone $\sigma$-complete effect algebras by effect-tribes, see e.g. \cite{BCD, Dvu1, Mun}.

We note that every monotone $\sigma$-complete effect algebra $E$ with an order-deter\-min\-ing system $\mathcal S$ of $\sigma$-additive states is representable. Indeed, for any $a\in E$, we define a function $\bar a:\mathcal S \to [0,1]$ such that $\bar a(s):=s(a)$, $s \in \mathcal S$. Then $\overline{E} :=\{\bar a: a \in E\}$ is an effect-tribe that is $\sigma$-isomorphic to $E$.

In addition, every monotone $\sigma$-complete effect algebra $E$ is representable if $E$ possesses an order-determining system of states: The family $\overline{E} :=\{\bar a: a \in E\}$ is a clan isomorphic to $E$ that is even an effect-tribe. Indeed, let $\{\bar a_n\}$ be a non-decreasing sequence of elements of $\overline{E}$. Then $\{a_n\}$ is a non-decreasing sequence of elements of $E$, so that there is $a\in E$ such that $\{a_n\}\nearrow a$. Then $\bar a_n\le \bar a$. If there is $b\in E$ such that $\bar a_n \le \bar b$ for each $n\ge 1$, then $a_n\le b$ and $a\le b$ which yields $\bar a\le \bar b$. Consequently, $E$ is $\sigma$-isomorphic to the effect-tribe $\overline{E}$.

In the next result we show that every representable monotone $\sigma$-complete effect algebra belongs to the class $\mathcal{OEP}(EA)$.

\begin{theorem}\label{th:3.3}
Every representable monotone $\sigma$-complete effect algebra satisfies {\rm OEP}.
\end{theorem}

\begin{proof}
Let  $\{x_t: t \in \mathbb R\}$ be  a system of elements of $E$ satisfying {\rm(3.2)--(3.4)}. Let $r_1,r_2,\ldots$ be any enumeration of the set of rational numbers. Due to the assumptions, there are a tribe $\mathcal T$ of fuzzy functions from $[0,1]^\Omega$ for some non-void set $\Omega$ and a $\sigma$-homomorphism $h$ from $\mathcal T$ onto $E$ such that if $f,g\in \mathcal T$ with $f\le g$ are given and if there is $c_1\in E$ with $h(f)\le c_1\le h(g)$, there is $c\in \mathcal T$ such that $f\le c\le g$ and $h(c)=c_1$. Given $r_n,$ let $a_n$ be a function from the tribe $\mathcal T$ such that $h(a_n)=x_{r_n}$ for any $n\ge 1$. We are claiming that it is possible to find such a sequence of functions $\{b_n\}$ from $\mathcal T$ such that $
h(b_n)=x_{r_n}$ for any $n\ge 1$ and $b_n \le b_m$ whenever $r_n < r_m.$ Indeed, if $n=1$, we set $b_1 = a_1.$ By mathematical induction suppose that we have find $b_1,\ldots,b_n$ such that $h(b_i)=x_{r_i}$ and $b_i \le b_j$ whenever $r_i < r_j$ for $i,j=1,\ldots, n$. Let $j_1,\ldots,j_n$ be a permutation of $1,\ldots,n$ such that $r_{j_1}<\cdots<r_{j_n}$.  For $r_{n+1}$ we have three possibilities (i) $r_{n+1}< r_{j_1}$, (ii) there exists $k =1,\ldots,n-1$ such that $r_{j_k} < r_{n+1} < r_{j_{k+1}},$ or (iii) $r_{j_n} < r_{n+1}$.  Applying the assumption, we can find $b_{n+1}\in \mathcal T,$ $h(b_{n+1})=r_{n+1},$ such that for all $i,j =1,\ldots, n+1$, $b_i \le b_j$ whenever $r_i < r_j$.

Thus, we can assume that the sequence of functions $\{b_{r_n}\}$, where $b_{r_n}:= b_n$ for $n \ge 1$,  is linearly ordered. Using the density of rational numbers in $\mathbb R$, for any real number $t\in \mathbb R$, we can find a function $b_t \in \mathcal T$ such that $h(b_t)=x_t$.  Indeed, if $\{p_n\} \nearrow t$ and $\{q_n\}\nearrow t$ for two sequences of rational numbers, $\{p_n\}$ and $\{q_n\}$, we can show that $\bigvee_n b_{p_n}=\bigvee_n b_{q_n}$.  Hence, $b_t := \bigvee_n b_{r_n}$ is a well-defined element of $\mathcal T$ satisfying $h(b_t)=r_t$.  In addition, the system of functions $\{b_t: t \in \mathbb R\}$ is also linearly ordered, and $b_t \le b_s$ if $t <s$.

Let $\omega$ be a fixed element of the set $\Omega$. We define $F_\omega(t):= b_t(\omega)$, $t \in \mathbb R$.  It is clear that $F_\omega$ is a nondecreasing, left continuous function such that $\lim_{t \to -\infty} F_\omega(t)=0$  and $\lim_{t \to \infty} F_\omega(t)=1$. By \cite[Thm 43.2]{Hal}, $F_\omega$ is a distribution function on the measurable space $(\mathbb R, \mathcal B(\mathbb R))$ corresponding to a unique probability measure $P_\omega$ on $\mathcal B(\mathbb R)$, that is $P_\omega((-\infty,t))=F_\omega(t)$ for every $t \in \mathbb R$. Define now a mapping $\xi: \mathcal B(\mathbb R) \to [0,1]^\Omega$ by $\xi(A)(\omega)=P_\omega(A)$, $A \in \mathcal B(\mathbb R)$, $\omega \in \Omega$.  In particular, we have $\xi((-\infty,t)) = b_t\in \mathcal T$ for any $t \in \mathbb R$.  To prove that every $\xi(A) \in \mathcal T$ for any $A \in \mathcal B(\mathbb R)$, let $\mathcal K=\{A \in \mathcal B(\mathbb R): \xi(A) \in \mathcal T\}$. Then $\mathcal K$ is a Dynkin system containing all intervals of the form $(-\infty, t)$ for $t \in \mathbb R$ which form a $\pi$-system.  Hence by  \cite[Thm 1.1]{Kal}, $\mathcal K$ is a $\sigma$-algebra of subsets, so that, $\mathcal K =\mathcal B(\mathbb R)$.

Therefore, $\xi$ is an observable on $\mathcal T$ and $x:=h\circ \xi$ is an observable on $M$ such that $x((-\infty, t)) = x_t$ for any $t \in \mathbb R$ which proves that $E$ satisfies OEP.
\end{proof}

\begin{example}\label{ex:3.4}
Let $\Omega$ be a non-empty set and $\mathcal D$ be a Dynkin system on $\Omega$. If $f:\Omega\to \mathbb R$ is a $\mathcal D$-measurable function, then the mapping $x_f:\mathcal B(\mathbb R) \to \mathcal D$ given by $x_f(A):=f^{-1}(A)$, $A \in \mathcal B(\mathbb R)$, is an observable on $\mathcal D$, in addition, $\mathcal S_f:=\{f^{-1}(A): A \in \mathcal B(\mathbb R)\}$ is a $\sigma$-algebra of subsets of $\Omega$.
If $s$ is a $\sigma$-additive state on $\mathcal D$, then $s_{x_f}(A):=s(f^{-1}(A))$, $A \in \mathcal B(\mathbb R)$, is a $\sigma$-additive measure on $ \mathcal B(\mathbb R)$.

In addition, $\mathcal D \in \mathcal{OEP}(EA)$.
\end{example}

\begin{proof}
Let $\{x_t: t \in \mathbb R\}$ be a system of subsets from $\mathcal D$ satisfying (3.2)--(3.4). Let $r_1,r_2,\ldots$ be any enumeration of the set real numbers of $\mathbb R$. Define a mapping $f: \Omega \to \mathbb R$ by
$f(\omega):= \inf\{r_j: \omega \in x_{r_j}\}$. Then $f^{-1}((-\infty, r_k))= \bigcup_{i: r_i<r_k} x_{r_i} = x_{r_k} \in \mathcal D$. So that, the set $\mathcal K=\{A \in \mathcal B(\mathbb R): f^{-1}(A) \in \mathcal D\}$ is a Dynkin system containing all intervals $(-\infty,r_k)$ which is a $\pi$-system. Hence, $\mathcal K=\mathcal B(\mathbb R)$, which proves $\mathcal D \in \mathcal{OEP}(EA)$.
\end{proof}

Let $x$ be a simple observable with $\sigma(x)=\{t_1,\ldots,t_n\}$ for some real numbers $t_1<\cdots <t_n$. Put $a_i=x(\{t_i\})$, $i=1,\ldots,n$, then $a_1+\cdots+a_n=1$, and for the spectral resolution of $x$, we have

$$
x((-\infty, t))= \left\{\begin{array}{ll} 0 & \mbox{if} \ t\le t_1,\\
a_1+\cdots + a_i & \mbox{if}\ t_i< t\le t_{i+1},\ i=1,\ldots,n-1,\\
1 & \mbox{if}\ t_n<t,
\end{array}
\right.
\eqno(3.5)
$$
for $t \in \mathbb R$. In particular, if $q_a$ is a question observable corresponding to an element $a\in E$, then the spectral resolution of $q_a$ is as follows

$$
q_a((-\infty, t))= \left\{\begin{array}{ll} 0 & \mbox{if} \ t\le 0,\\
a' & \mbox{if}\ 0< t\le 1,\\
1 & \mbox{if}\ 1<t,
\end{array}
\right.
\eqno(3.6)
$$
for $t \in \mathbb R$.

\section{Sharp Elements and Observables on Orthoalgebras}

We show that the set of sharp elements of a homogeneous monotone $\sigma$-complete effect algebra is a monotone $\sigma$-complete sub-effect algebra. It will describe the so-called orthoalgebraic skeleton. It allows to characterize sharp observables by a sharp spectral resolution. In addition, we show that every monotone $\sigma$-complete orthoalgebra with RIP satisfies the Observable Existence Property.

In \cite[Cor 4.4]{Jen}, there was proved that if $E$ is a homogeneous effect algebra, then $\Sh(E)$ is a sub-effect algebra of $E$. Now we extend this result for monotone $\sigma$-complete homogeneous effect algebras.

\begin{theorem}\label{th:3.5}
If $E$ is a monotone $\sigma$-complete homogeneous effect algebra, then $\Sh(E)$ is a monotone $\sigma$-complete homogenous sub-effect algebra of $E$.
\end{theorem}

\begin{proof}
Let $\{b_n\}$ be a sequence of sharp elements of $E$  such that $\{b_n\}\nearrow b$, where $b \in E$. We have to show that $b$ is a sharp element of $E$. We define elements $a_1=b_1$, and $a_n= b_n-b_{n-1}$ for $n\ge 2$. Then the sequence $\{a_n\}$ is summable with sum $b=\sum_n a_n$. The set $M=\{b', a_1,a_2,\ldots\}$ is an internally compatible subset of $E$, so there is a block $B$ of $E$ containing $M$. Since $B$ is a sub-effect algebra of $E$ with RDP, we get $b \in B$ and every $b_n \in B$. Hence, $b=\bigvee_B b_n= \bigvee_n b_n$, where $\bigvee_B$ is the join taken in $B$. Since each $b_n$ is a sharp element of $E$, so is sharp in the block $B$ which satisfies RDP. This implies that every $b_n$ is a central element of the effect algebra $B$. Therefore, $b_n\wedge_B x$ exists in $B$ for each $x \in B$; here $\wedge_B$ means the meet taken in the block $B$, in contrast to $\wedge$, $\vee$  which denote meet and join taken in the whole $E$.

Given a subset $M$ of $E$, we define inductively a sequence $\{M_n:n\ge 0\}$ of subsets of $E$ as follows: $M_0:=M$, and for each $n\ge 1$, we put
$M_n=\{x\in E: x\le y,y' \mbox{ for some } y \in M_n\}\cup \{y-x: x \le y,y' \mbox{ for some } y \in M_n \}$, and for them we put $\overline{M}=\bigcup_{n=0}^\infty M_n$. If $M$ is a block of $E$, by \cite[Cor 3.8(b)]{Jen}, $M=\overline{M}$.

Let $x \le b,b'$. Since $B=\overline{B}$, $x\in B$. First we show that $b_n\wedge x=0=b_n\wedge_B x$ for each $n$. Let $z\le x,b_n$, then $z\le x\le b' \le b_n'$ so that $x=0$. We assert, that there is $x_0 \in B$ such that $b_n \wedge_B x \le x_0\le x,b$ for each $n$. At any rate, the element $x_0=0\in B$ is such an element. Now let $x_0$ be any element of $B$ such that $0=b_n\wedge_B x\le x_0\le x,b$.

Now we prove a series of Claims.

\vspace{2mm}\noindent
{\it Claim 1.} $x = (x - x_0) +x_0$, $b =
(b- x_0) +x_0,$ $(x-x_0) + (b- x_0) + x_0
\in E.$
\vspace{2mm}

Indeed, we have $x-(b_n\wedge_B x)=(b_n\vee_B x)-b_n$. Since $1=(b_n\vee_B x)'+((b_n\vee_B x)-b_n)+b_n$, we have $b_n \le (x-(b_n\wedge_B x))'\le (x-x_0)'$ and $b\le (x-x_0)'$ which gives $(x-x_0)+b \in E$. Clearly $x= (x-x_0)+x_0$, $b = (b-x_0)+x_0$. In addition, $ x,b \le (x-x_0)+b$.

\vspace{2mm}\noindent
{\it Claim 2.} $x_0=0$.
\vspace{2mm}

To prove the claim, it is sufficient to show that $\bigwedge_{B,n} (x -(b_n \wedge_B x))=x-x_0= \bigwedge_n (x -(b_n\wedge_B x))$.
It is clear that $x-x_0 \le x-(x\wedge_B b_n)$ for each $n\ge 1$. Let $d\in E$ be such that $d\le x-(x\wedge_B b_n)$ for each $n$. Therefore, $d \le b,b'$ so that $b \in B$.
Then $d\le x-(x\wedge_B b_n)= (x\vee_B b_n)-b_n$ and $d+b_n \le x\vee_B b_n \le (x-x_0)+b$ when we have used the end of the proof of Claim 1. Whence, $ b_n \le ((x-x_0)+b)-d$ and $b \le ((x-x_0)+b)-d$ which gives $b+d\le (x-x_0)+b$ and $d\le x-x_0$. Finally $x-x_0=\bigwedge_{B,n} (x -(x\wedge_B b_n))=\bigwedge_n (x -(b_n\wedge_B x))=x$ which gives $x_0=0$.

\vspace{2mm}\noindent
{\it Claim 3.} $x\wedge b =0= x\wedge_B b$.
\vspace{2mm}

Let for $z \in E$, we have $z \le x,b$. Since $z\le x\le b'$, we have $z \le b,b'$ and $b \in B$. Then $0=b_n\wedge_B x\le z \le x,b$ which by Claim 2 yields $z=0$.

Finally, since we have $x\le b,b'$, then by Claim 3, we have $x=x\wedge b=0$ and this proves that $b$ is a sharp element of $E$, and consequently, $\Sh(E)$ is a monotone $\sigma$-complete sub-effect algebra which is homogenous because $\Sh(E)$ is an sub-orthoalgebra of $E$ and every orthoalgebra is homogenous.
\end{proof}

\begin{remark}\label{re:3.6}
{\rm We note that if $E$ is a homogeneous effect algebra, then $\Sh(E)$ is the biggest orthoalgebra that is a sub-effect algebra of $E$ and every sharp element in the sub-effect algebra is a sharp element of $E$.  Indeed, if $E_0$ is any orthoalgebra that is a sub-effect algebra of $E$ such that $\Sh(E_0)\subseteq \Sh(E)$, then for each $x\in E_0$, we have $x\wedge x'$ exists in $E$ and is zero, so that $E_0\subseteq \Sh(E)$. We call the biggest orthoalgebra of $E$  the {\it orthoalgebraic skeleton} of $E$ if it exists in $E$. Under the conditions of Theorem \ref{th:3.5}, $\Sh(E)$ is a monotone $\sigma$-complete orthoalgebraic skeleton of $E$.

We note, as it was already said, that $\mathcal P(H)$ is the orthoalgebraic skeleton of $\mathcal E(H)$.
}
\end{remark}

Now we characterize sharp observables as follows.

\begin{proposition}\label{pr:4.11}
Let $x$ be an observable on a homogenous monotone $\sigma$-complete effect algebra $E$. The following statements are equivalent:

\begin{enumerate}

\item[{\rm (i)}] $x$ is a sharp observable.

\item[{\rm (ii)}] $x((-\infty,t)) \in \Sh(E)$ for each $t \in \mathbb R$.

\item[{\rm (iii)}] $x((-\infty,t)) \in \Sh(E)$ for each $t \in \mathbb Q$.
\end{enumerate}
\end{proposition}

\begin{proof}
Of course (i) implies (ii), and (ii) does (iii).

(ii) $\Rightarrow$ (i). Let $\mathcal K:=\{A \in \mathcal B(\mathbb R): x(A) \in \Sh(E)\}$. We assert that $\mathcal K$ is a Dynkin system. Assume $A,B \in \mathcal K$, $A\cap B = \emptyset$. Let $a \in E$ be such that $a \le x(A\cup B), x((A\cup B)^c)$. Due to homogeneity of $E$, there are $a_1,a_2\in E$ with $a_1 \le x(A)$ and $a_2\le x(B)$ such that $a=a_1 +a_2$. Then $a_1\le a_1+a_2\le (x(A)+x(B))'\le x(A)'$ which yields $a_1=0$, and in a similar way, we have $a_2=0$. Hence, $a=0$, $x(A\cup B) \in \Sh(E)$, and $A\cup B \in \mathcal K$.

Now let $\{A_i\}$ be a sequence of mutually disjoint Borel subsets from $\mathcal K$. By the first part of the present proof, $B_n=A_1\cup \cdots \cup A_n \in \mathcal K$. Applying Theorem \ref{th:3.5}, we see that $A = \bigcup_n A_n = \bigcup_n B_n \in \mathcal K$.

Hence, $\mathcal K$ is clearly a Dynkin system containing all intervals of the form $(-\infty, t)$, $t \in \mathbb R$. These intervals form a $\pi$-system, hence, by the Sierpi\'nski Theorem, \cite[Thm 1.1]{Kal}, $\mathcal K$ is a $\sigma$-algebra, which proves $\mathcal K= \mathcal B(\mathbb R)$.

(iii) $\Rightarrow$ (ii). Let $t\in \mathbb R$. There is a sequence $\{r_n\}$ of rational numbers such that $\{r_n\}\nearrow t$. Applying Theorem \ref{th:3.5}, we see that $x((-\infty,t))\in \Sh(E)$.
\end{proof}

This allows us to show that any spectral resolution consisting of sharp elements of a monotone $\sigma$-complete homogeneous effect algebra always entails the corresponding sharp observable.

If $E_1$ is a sub-orthoalgebra of an orthoalgebra $E$ such that $E_1$ is a Boolean algebra in its own right, we call $E_1$ a {\it Boolean sub-orthoalgebra} of $E$. If in addition, $E_1$ is a Boolean $\sigma$-algebra in its own right, we call $E_1$ a {\it Boolean $\sigma$-sub-orthoalgebra} of $E$.

\begin{theorem}\label{th:3.7}
Let $x$ be an observable of a monotone $\sigma$-complete orthoalgebra $E$. Then the range $\mathcal R(x)$ is a Boolean $\sigma$-sub-orthoalgebra of $E$.
\end{theorem}

\begin{proof}
(i) The range $\mathcal R(x)$ contains $0,1$ and is closed under the operation $'$.  Now let $x(G)+x(F)$ exist in $E$. Then $x(G)+x(F)=x(G\setminus F)+x(G\cap F)+x(F\setminus G)+x(G\cap F)$ which means $x(G\cap F)=0$. Therefore, $x(G)+x(F)=x((G\setminus F)\cup (F\setminus G))=x(G\cup F) \in \mathcal R(x)$ which shows that the range $\mathcal R(x)$ is a sub-orthoalgebra of $E$.

In particular, if $x(A)\le x(B)$, then $x(B)-x(A)=(x(B)'+x(A))'=x((B'\cup A)')=x(B\setminus A)$. Hence, $x(A\cup B)= x((A\cup B)\cap A')+x((A\cup B)\cap A)=x(B\setminus A)+x(A)=x(B)$. In a similar way, we can show $x(A)=x(A\cap B)$.

(ii) Now we show that $\mathcal R(x)$ is an orthomodular poset in its own right. To show that we have to verify that if $x(A)+x(B)$, $x(A)+x(C)$, $x(B)+x(C)$ exist in $\mathcal R(x)$ then $x(A)+x(B)+x(C)$ is defined in $\mathcal R(x)$. From the first part of the proof we have $x(A\cap B)=x(A\cap C)=x(B\cap C)=0$. Hence, $x(A)=x(A\setminus (B\cup C))$, $x(B)=x(B\setminus (A\cup C))$ and $x(C)=x(C \setminus (A\cup B))$ which yields $x(A\cup B \cup C)=x(A)+x(B)+x(C)\in \mathcal R(x)$. Then $a+b=a\vee_{\mathcal R(x)} b$ for $a,b \in \mathcal R(x)$, where $\vee_{\mathcal R(x)}$ is the join taken in $\mathcal R(x)$.

(iii) We establish that $\mathcal R(x)$ is monotone $\sigma$-complete. To show that, let $\{x(G_i)\}$ be a sequence of summable elements. According to (ii), we can show that $x(G_1)+\cdots +x(G_n)=x(G_1\cup\cdots \cup G_n)$. Hence, if $G =\bigcup_n G_n$, then $x(G) = \bigvee_n x(G_1\cup\cdots \cup G_n)=\sum_n x(G_n)$.

(iv) In this part we establish that $\mathcal R(x)$ satisfies RDP. Let $x(A) \le x(F)+x(G)$. Then $x(A)\le x(A)+x(F)= x((F\setminus G)\cup (G\setminus F))$. Hence by the end of (i), $x(A) = x(((F\setminus G)\cup (G\setminus F))\cap A)= x((F\setminus G)\cap A)\cup ((G\setminus F)\cap A))= x((F\setminus G)\cap A)+ x((G\setminus F)\cap A)$.

Finally, we have established that $\mathcal R(x)$ is a monotone $\sigma$-complete orthoalgebra with RDP, therefore, the set of sharp elements of $\mathcal R(x)$ coincides with the set of central elements of $\mathcal R(x)$, see \cite[Thm 3.2, Thm 5.11]{Dvu2}, but $\Sh(\mathcal R(x))=\mathcal R(x)$, so that $\mathcal R(x)$ is a Boolean $\sigma$-algebra in its own right.
\end{proof}

As we have seen, the range of any observable on a monotone $\sigma$-complete orthoalgebra $E$ is always a Boolean sub-orthoalgebra of $E$ that is a sub-orthoalgebra of $E$. For example, if $E$ is an orthoalgebra, given $a\in E$, the set $\{0,a,a',1\}$ is always a Boolean subalgebra of $E$, hence $E$ can be covered by a system of Boolean sub-orthoalgebras of $E$, see also \cite[Lem 3.2]{FGR}. In addition, if $E$ is a monotone $\sigma$-complete orthoalgebra, and $x$ is an observable defined by (3.0), then $E$ can be covered by a system of Boolean $\sigma$-sub-orthoalgebras of $E$.

\begin{lemma}\label{le:3.8}
{\rm (1)} Let $B$ and $B_1$ be two Boolean sub-orthoalgebras of an orthoalgebra $E$ such that $B \subseteq B_1$. If $a,b \in B$, then $a\vee_{B}b=a\vee_{B_1} b$ and $a\wedge_{B}b=a\wedge_{B_1} b$.

{\rm (2)} Any orthoalgebra $E$ that is not a Boolean algebra can be covered by the system of maximal Boolean sub-orthoalgebras of $E$.
\end{lemma}

\begin{proof}
(1) Let $B \subseteq B_1$ and let $a,b \in B$. Since $B$ is in fact a sub-orthoalgebra of $E$ such that if $x+y$ exists for $x,y \in B$, then $x+y = x\vee_B y$. In addition, RDP holds in $B$ because $B$ is a Boolean algebra. Since $a+a'=b+b'$, there are four elements $c_{11},c_{12},c_{21},c_{22}\in B$ such that $a = c_{11}+c_{12},$ $a'= c_{21}+c_{22},$ $b= c_{11} + c_{21}$ and $b'= c_{12}+c_{22}$. Then the element $c:=c_{11}+c_{12}+c_{21}$ exists in $B$ so that $c=a\vee_B c_{21}=b\vee_B c_{12}$. If $x\in B$ is an upper bound for $a,b$, then $x\ge a,b,c_{12},c_{21}$ so that $x\ge c$ which proves $c= a\vee_B b$. Since the elements $c_{11},c_{12},c_{21},c_{22}$ are also in $B_1$, we have $c \in B_1$ and $c=a\vee_{B_1} b$.

(2) Let $x \in E$ be an arbitrary element and let $\{B_t:t \in T\}$ be a chain of Boolean sub-orthoalgebras of $E$ containing $x$, where each $B_t\ne E$; such Boolean sub-orthoalgebras exist because $\{0,x,x',1\}$ is such a Boolean sub-orthoalgebra. If we set $B_0= \bigcup_t B_t$, then $B_0$ is clearly a sub-orthoalgebra of $E$. If $a,b \in B_0$, then there is $t_0$ such that $a,b\in B_{t_0}$, and let $c=a\vee_{B_{t_0}} b$. We assert that $c$ is the join in the whole $B_0$. Indeed, let $x \in B_0$ be an upper bound of $a,b$. There is $t_1$ such that $a,b,x\in B_{t_1}$. By the first part of the present proof, $c=a\vee_{B_{t_1}} b\le x$ and $c=a\vee _{B_0} b$. Hence, $B_0$ is a Boolean algebra and $B_0\ne E$ because $E$ is not a Boolean algebra.
Applying Zorn's lemma, we have that there is a maximal Boolean sub-orthoalgebra $E$. Consequently, $E$ can be covered by the set of maximal Boolean sub-orthoalgebras of $E$.
\end{proof}

Maximal Boolean sub-orthoalgebras of an orthoalgebra in which RDP holds are maximal sub-orthoalgebras with RDP and vice-versa. Indeed, every element of a maximal sub-orthoalgebra with RDP is central, so that,  blocks are maximal Boolean sub-orthoalgebras of the orthoalgebra. In other words, maximal Boolean sub-orthoalgebras of orthoalgebras play the same role as blocks do in homogeneous effect algebras.

For example, the well-known orthoalgebra, the Fano plane, see e.g. \cite[p. 101]{DvPu}, consists of seven Boolean sub-orthoalgebras with eight elements; they are only blocks.

It is worthy of recalling that if $B$ is any Boolean sub-orthoalgebra of an orthoalgebra $E$, then (i) for each $a,b\in B$, $a\leftrightarrow b$, (ii) if $a,b\in B$, $a=a_1+c$, $b=b_1+c$ and $a_1+b_1+c$ exists in $E$, where $a_1,b_1,c\in B$, then $a\vee _B b= a_1+b_1+c$, $a\wedge_B b=c$, and $a_1\wedge_B b_1=0$.

\begin{lemma}\label{le:3.9}
{\rm (1)} Let $B$ be a Boolean sub-orthoalgebra of an orthoalgebra $E$ satisfying \RIP. Then any two elements $a,b\in B$ are strongly compatible, $a\vee b$ exists in $E$ and $a\vee b=a\vee_B b$, and $a\wedge b= a\wedge_B b$.

{\rm (2)} If $a,b$ are arbitrary elements of any orthoalgebra $E$ with {\rm RIP} such that $a+b$ exists in $E$. Then $a$ and $b$ are strongly compatible, $a\vee b$ exists, and $a\vee b =a+b$.
\end{lemma}

\begin{proof}
(1) As it was just said above, every two elements of $B$ are compatible.

(i) Assume that $a+b$ is defined in $E$. Then $a\vee_B b= a+b$. Since $a$ and $b$ are compatible in $B$, $a=a_1+c$, $b=b_1+c$, and $a_1+b_1+c \in B$ for some $a_1,b_1,c \in B$. Then $a_1+b_1+c=a\vee_B b = a+b$, so that $c=0$. Let for some $x \in E$, we have $x\le a,b$. Then $x\le a\le b'$ gives $x=0$ which implies $a_1\wedge b_1$ exists in $E$ and $a\wedge b=0$. Applying \cite[Prop 2.2]{CIJTP}, we have that $a\vee b$ exists in $E$ and $a\vee b=a+b$. But $a\vee _B b = a+b$.

(ii) Now let $a,b\in B$ be arbitrary, then $a=a_1+c$, $b=b_1+c$, $a_1+b_1+c= a\vee_B b$. By the first part of the present proof, $a\vee b_1=a_1+b_1+c= b\vee a_1$ which yields $a\vee b$ exists in $E$, and $a\vee b=a_1+b_1+c=a\vee _B b$. Because $a$ and $b$ are arbitrary elements of $B$, then $a'\vee b'$ exists in $B$ so that $a\wedge b=(a'\vee b')'$ is defined in $B$.
This yields that $a\wedge b$ and $a\vee b$ are defined in $B$ and they coincide with $a\wedge_B b$ and $a\vee_B b$, respectively. This in particular implies that if $a=a_1+c$, $b=b_1+ c$ and $a_1+b_1+c$ exists in $B$ for some $a_1,b_1,c \in B$, then $a_1\wedge b_1 = a_1\wedge_B b_1 =0$ which implies $a$ and $b$ are strongly compatible.

(2) Set $a_1=a$, $a_2=b$, $a_3=(a+b)'$ and define an observable $x$ via (3.0). Then $B=\mathcal R(x)$ is by Theorem \ref{th:3.7} a Boolean sub-orthoalgebra of $E$. Hence, by (1), $a$ and $b$ are strongly compatible, so that $a\vee_B b = a+b = a\vee b$.
\end{proof}

The notion of an observable was defined as a mapping $x$ from $\mathcal B(\mathbb R)$ into a monotone $\sigma$-complete effect $E$. We can extend this notion also for an arbitrary effect algebra $E$ assuming that $x:\mathcal B(\mathbb R)\to E$ such that (i) $x(\mathbb R)=1$, (ii) $x(A\cup B)=x(A)+x(B)$ whenever $A,B \in \mathcal B(\mathbb R)$ are disjoint, and (iii) if $\{A_i\}\nearrow A$, then $\bigvee_i x(A_i)$ is defined in $E$ and $x(A)=\bigvee_i x(A_i)$. A mapping $x:\mathcal B(\mathbb R)\to E$ satisfying (i) and (ii) is said to be an {\it f-observable} (f stands for finitely additive). It has analogous properties as do observables stated in the beginning of the third section just after definition of an observable besides (vi).

\begin{lemma}\label{le:3.10}
Let $E$ be an orthoalgebra with {\rm RIP} and let $a,b \in E$. The following statements are equivalent:
\begin{enumerate}

\item[{\rm (i)}] The elements $a$ and $b$ are strongly compatible.

\item[{\rm (ii)}] The elements $a$ and $b$ are  compatible.

\item[{\rm (iii)}]
There is an f-observable $x$ of $E$ and two Borel sets $A$ and $B$ of $\mathbb R$ such that $a=x(A)$ and $b= x(B)$.

\item[{\rm (iv)}] There is a Boolean sub-orthoalgebra of $E$ containing $a$ and $b$.
\end{enumerate}
\end{lemma}

\begin{proof}
Evidently (i) implies (ii).

(ii) $\Rightarrow$ (iii). There are elements $a_1,b_1,c \in E$ such that $a=a_1+c$, $b=b_1+c$, and $a_1+b_1 +c$ is defined in $E$. 
Take summable elements $a_1,b_1,c,(a_1+ b_1+c)'$ and define an f-observable $x$ via (3.0). Then $a,b \in \mathcal R(x)$.

(iii) $\Rightarrow$ (iv). This follows easily from Theorem \ref{th:3.7}.

(iv) $\Rightarrow$ (i). Let $B$ be a Boolean sub-orthoalgebra of $E$ containing $a$ and $b$. By Lemma \ref{le:3.9}, $a$ and $b$ are strongly compatible.
\end{proof}

\begin{lemma}\label{le:3.11}
Let $a$ be an element of an orthoalgebra $E$ satisfying \RIP, and let $B(a)=\{b \in E: b \stackrel{\mbox{\rm c}}{\longleftrightarrow}a\}$. Then $B(a)$ is a sub-orthoalgebra of $E$ that is an orthomodular poset.

If, in addition, $E$ is monotone $\sigma$-complete, then $B(a)$ is also monotone $\sigma$-complete.

Every orthoalgebra with {\rm RIP} is an orthomodular poset.
\end{lemma}

\begin{proof}
It is clear that  $0,1, a \in B(a)$.

We show that if $b \in B(a)$, then $b'\in B(a)$. This follows from the criterion  (ii) of Theorem \ref{le:3.10}.

Assume $b_1,b_2 \in B(a)$ and let $b_1+b_2$ exist in $E$. By Lemma \ref{le:3.9}(2), $b_1$ and $b_2$ are strongly compatible, and $b_1\vee b_2=b_1+b_2$. Applying \cite[Prop 2.9]{CIJTP}, $a$ is strongly compatible with $b_1\vee b_2=b_1+b_2$.

To show that $B(a)$ is an OMP in its own right, we use the criterion \cite[Thm 2.12]{FGR} (see also \cite[Prop 1.5.6]{DvPu}) saying that an orthoalgebra is an OMP iff the existence of $b+c$ implies $b\vee a$ exists.
Thus let $b,c \in B(a)$ and let $b+c\in E$. As it was already said, Lemma \ref{le:3.9}(2), $b,c$ are strongly compatible, so that $b\vee c$ exists in $B(a)$, and $b\vee c=b+c$.

Let $E$ be monotone $\sigma$-complete and let $\{b_n\}$ be a non-decreasing sequence of elements of $B(a)$, and let $b =\bigvee_n b_n$.  By \cite[Prop 3.2]{270}, $b \stackrel{\mbox{\rm c}}{\longleftrightarrow}a$, proving $b\in B(a)$.

To establish that $E$ is an OMP, we use the fact that if we set $a=0$, then $B(a)= E$, and therefore, as it was just proved, $E$ is an OMP.
\end{proof}

We remark that an effect algebra with RIP is not necessarily an orthomodular poset. Indeed, if an effect algebra $E$ has RDP (so it has RIP) and by \cite{Rav}, $E$ is isomorphic to an interval effect algebra $\Gamma(G,u)$ and these effect algebras are not orthomodular posets, in general.

The following result follows directly from Lemma \ref{le:3.11}.

\begin{theorem}\label{th:3.12}
Every monotone $\sigma$-complete orthoalgebra with {\rm RIP} satisfies {\rm OEP}.
\end{theorem}

\begin{proof}
Let $E$ be a monotone $\sigma$-complete orthoalgebra with RIP. By Lemma \ref{le:3.11}, $E$ is a monotone $\sigma$-complete orthomodular poset, equivalently, $E$ is a $\sigma$-orthocomplete orthomodular poset. In view of \cite[Thm 3.8]{DvKu}, $E$ has OEP.
\end{proof}

We note that that there are OMPs where RIP fails.  Indeed, let $\Omega=\{1,\ldots,8\}$ and let $E$ be the system of all subsets of $\Omega$ with even number of elements. Then $E$ is an OMP and since $\{1,2\},\{1,3\} \subseteq \{1,2,3,4\}, \{1,2,3,5\}$, RIP fails. On the other side, every finite OMP with RIP is a lattice, and there are also  infinitely-countable  OMP's with RIP that are not lattices, see \cite[Thm 2.3(ii)]{DvPt}.

The question whether every monotone $\sigma$-complete orthoalgebra belongs to the class $\mathcal{OEP}(EA)$ is still open. If each block of a monotone $\sigma$-complete orthoalgebra would be a monotone $\sigma$-complete sub-orthoalgebra, then the answer to the first question would be positive. Unfortunately, also the answer to the second question is unknown, see also \cite[Sec 3, p. 3311]{Sch}.

\section{Conclusion}

The Riesz Decomposition Property is an important notion that allows us to show that effect algebras with RDP are in fact interval ones in some unital Abelian po-groups with interpolation. A more general notion are homogeneous effect algebras.

Observables model quantum mechanical measurements. They are defined as $\sigma$-homomorphisms from the set of Borel sets on $\mathbb R$ preserving order, negations, addition, and monotone limits. We have studied a problem when a given system of elements, called a spectral resolution, of a monotone $\sigma$-complete effect algebra implies the existence of the corresponding observable. Theorem \ref{th:3.3} says that every representable monotone $\sigma$-complete effect algebra has this property, and Theorem \ref{th:3.12} describes an analogous result for monotone $\sigma$-complete effect algebras with the Riesz Interpolation Property.

In Theorem \ref{th:3.5} we showed that the set of sharp elements of a homogeneous monotone $\sigma$-complete effect algebra is always a monotone $\sigma$-complete sub-effect algebra. This result was applied in Proposition \ref{pr:4.11} to sharp observables showing their characterization by the corresponding spectral resolution whose elements are only sharp elements. In addition, in Lemma \ref{le:3.11} we showed that every orthoalgebra with RIP is in fact an OMP.

\end{document}